\documentclass[reqno,12pt]{amsart}
\usepackage{fullpage}
\usepackage{graphicx}
\usepackage{caption}
\usepackage{subcaption}
\usepackage{amsfonts}
\usepackage{amssymb}
\usepackage{amsmath}
\usepackage{enumerate}

\def\Im{{\rm Im}\,}
\def\sgn{{\rm sgn}}
\def\Rnum{\mathbb{R}}

\def\const{\text{const}}

\def\txtint{{\textstyle\int}}

\newtheorem{prop}{Proposition}
\newtheorem{thm}{Theorem}
\newtheorem{cor}{Corollary}
\newtheorem{lem}{Lemma}

\def\Ref#1{Ref.~\cite{#1}}
\def\secref#1{Sec.~\ref{#1}}

\def\scrpt#1{$\scriptstyle {#1}$}

\begin{document}
\allowdisplaybreaks[3]

\title{Accelerating dynamical peakons and their behaviour}

\author{
Stephen C. Anco$^1$
\\\lowercase{\scshape{ and }}\\
Elena Recio$^2$ 
\\
\\\lowercase{\scshape{
${}^1$Department of Mathematics and Statistics\\
Brock University\\
St. Catharines, ON L\scrpt2S\scrpt3A\scrpt1, Canada}} \\\\
\lowercase{\scshape{
${}^2$Department of Mathematics\\
Faculty of Sciences, Universidad de C\'adiz\\
Puerto Real, C\'adiz, Spain, 11510}}\\
}

\thanks{sanco@brocku.ca, elena.recio@uca.es.}

\begin{abstract}
A wide class of nonlinear dispersive wave equations are shown to possess 
a novel type of peakon solution in which the amplitude and speed of the peakon are time-dependent. 
These novel dynamical peakons exhibit a wide variety of different behaviours 
for their amplitude, speed, and acceleration, 
including an oscillatory amplitude and constant speed which describes a peakon breather.
Examples are presented of families of nonlinear dispersive wave equations that illustrate 
various interesting behaviours, such as 
asymptotic travelling-wave peakons,  
dissipating/anti-dissipating peakons, 
direction-reversing peakons, 
runaway and blow up peakons,
among others. 
\end{abstract}

\maketitle

\section{Introduction}

Peakons are peaked travelling waves of the form $u(x,t)=a\exp(-|x-ct|)$
which were first found as weak solutions for the Camassa-Holm (CH) equation \cite{CamHol}
$u_t -u_{txx} + 3uu_x -2u_xu_{xx} -uu_{xxx} =0$,
with the speed-amplitude relation being $c=a$. 
Several other similar peakon equations are well known:
Degasperis-Procesi (DP) equation \cite{DegPro,DegHolHon}
 $u_t -u_{txx} + 4 uu_x -3u_xu_{xx} -uu_{xxx} =0$,
with $c=a$; 
Novikov (N) equation \cite{Nov}
$u_t -u_{txx} + 4u^2u_x -3uu_xu_{xx} -u^2u_{xxx} =0$,
with $c=a^2$;
modified Camassa-Holm (mCH) equation (also known as FORQ equation) \cite{Fok95a,Olv-Ros,Fok-Olv-Ros,Fuc}
$u_t -u_{txx} + 3u^2u_x -u_x^3-4uu_xu_{xx} +2u_xu_{xx}^2 -(u^2-u_x^2)u_{xxx} =0$, 
with $c=\tfrac{2}{3}a^2$.
Much of the interest in these equations is that, 
firstly, they are integrable systems 
having a Lax pair, bi-Hamiltonian structure, hierarchies of symmetries and conservation laws; 
secondly, 
the peakon solutions are orbitally stable
\cite{ConStr,Len,LiuLin},
which implies the shape of the peakon is unchanged under small perturbations; 
thirdly, 
they possess $N$-peakon weak solutions given by 
a linear superposition of single peakons with time-dependent amplitudes and speeds;
and fourthly, 
they exhibit wave breaking 
in which certain smooth initial data yields solutions whose gradient $u_x$ 
blows up in a finite time while $u$ stays bounded
\cite{ConEsc1998a,ConEsc1998b,EscLiuYin,LiuOlvQuZha,GuiLiuOlvQu,JiaNi}.
Moreover, the CH and DP equations arise as models for water waves
\cite{Joh,ConLan,Ion-Kru}, 
and it is known that the travelling wave solutions of greatest height 
for the Euler equations governing water waves 
have a peak at their crest
(see \cite{Con06,Con12,ConEsc07,Tol}). 

All of these equations, 
and their various modified versions and nonlinear generalizations \cite{HolHon,MiMu,GraHim,AncSilFre,HimMan2016a,HimMan2016b,AncRec}, 
belong to the general family of nonlinear dispersive wave equations 
\begin{equation}\label{fg-fam}
m_t + f(u,u_x)m +(g(u,u_x)m)_x=0, 
\quad 
m=u-u_{xx}  
\end{equation} 
where $f$ and $g$ are arbitrary non-singular functions of $u$ and $u_x$. 
Remarkably, as shown in recent work \cite{RecAnc2016}, 
every equation in this family \eqref{fg-fam}
possesses $N$-peakon weak solutions 
\begin{equation}\label{Npeakon}
u(x,t)=\sum_{i=1}^{N} a_i(t) \exp(-|x-x_i(t)|), 
\quad 
N\geq 1
\end{equation}
whose time-dependent amplitudes $a_i(t)$ and positions $x_i(t)$
satisfy a nonlinear dynamical system consisting of $2N$ coupled ODEs 
\begin{equation}\label{fgpeakonsys}
\dot a_i = \tfrac{1}{2} [F(u,u_x)]_{x_i},
\quad
\dot x_i = -\tfrac{1}{2} [G(u,u_x)]_{x_i}/a_i
\end{equation}
in terms of $F=\int f(u,u_x)\,d u_x$ and $G=\int g(u,u_x)\,d u_x$,
where square brackets denote the jump at a discontinuity.
In the case $N=1$, 
single peakon solutions are travelling waves $u(x,t)=a\exp(|x-ct|)$ with $c\neq 0$ 
if and only if $F$ and $G$ have the properties \cite{RecAnc2016}
\begin{equation}\label{FG-cond}
F(a,a)=F(a,-a),
\quad
G(a,a)\neq G(a,-a) 
\end{equation}
for arbitrary $a\in\Rnum$. 

These results have two important consequences: 
multi-peakons exist for nonlinear dispersive wave equations \eqref{fg-fam}
without the need for any integrability properties or any Hamiltonian structure,
while single peakon travelling waves exist only when the nonlinearities $f$ and $g$ in the wave equation satisfy certain conditions. 
This is a sharp contrast to the situation for soliton solutions of nonlinear dispersive evolution equations, 
where solitary travelling waves in general exist without any restrictions on the nonlinearity in the equation, 
while the existence of $N$-soliton solutions for arbitrary $N>2$ usually requires that 
the equation be integrable. 

The purpose of the present work is to study novel $1$-peakon solutions 
$u(x,t)=A(t) \exp(|x-X(t)|)$ 
that are more general than travelling waves. 
We will refer to these solutions as \emph{dynamical peakons}. 
They arise whenever the nonlinearities $f$ and $g$ in a nonlinear dispersive wave equation \eqref{fg-fam}
fail to satisfy the necessary and sufficient conditions \eqref{FG-cond}
for a $1$-peakon to have a constant amplitude and a constant non-zero speed. 

Dynamical peakons are a new nonlinear phenomena that have not been previously recognized to exist.

In general, 
dynamical peakons have a time-dependent amplitude, $\dot A\neq 0$,
and either a time-dependent speed, $\ddot X\neq 0$,
or a constant speed, $\dot X=\const$. 
Their amplitude, speed, and acceleration 
can exhibit a wide variety of different behaviours, 
including an oscillatory amplitude and constant speed which describes a peakon breather.
Other behaviours for the amplitude are 
a finite-time blow-up or extinction, 
and a long-time unboundedness, or extinction, or finite asymptote. 
For the speed and acceleration, novel behaviours are 
a finite-time braking or runaway or wheelspin-limit, 
and a long-time braking, runaway, or finite asymptotic limit,
as well as direction reversal.

In \secref{conditions},
we write out the general class of nonlinearities $f$ and $g$ 
for which dynamical peakons exist, 
and we determine the special class of these nonlinearities 
such that dynamical peakons are accelerating. 
We also discuss the connection of these conditions to the absence of conservation laws 
for momentum and the Sobolev norm. 
In \secref{results},
we state conditions on the nonlinearities $f$ and $g$ that yield 
various interesting types of behaviour for the amplitude and speed of accelerating dynamical peakons. 
We illustrate these conditions by classifying the behaviour of all dynamical peakons 
for a family of wave equations in which the nonlinearities $f$ and $g$ are powers of $u$. 
In \secref{examples}, 
we give explicit examples of asymptotic travelling-wave peakons, 
dissipating/anti-dissipating peakons, blowing-up peakons, direction reversing peakons, 
and peakon breathers. 
We make some concluding remarks in \secref{remarks}.

\section{Dynamical peaked waves}
\label{conditions}

A dynamical peakon is the $N=1$ case of the $N$-peakon weak solution \eqref{Npeakon},
which we will write in the form 
\begin{equation}\label{1peakon}
u(x,t)= A(t) e^{-|x-X(t)|} . 
\end{equation}
Physically, this type of peakon describes a dynamically evolving peaked wave,
with amplitude $A(t)$ and position $X(t)$. 

For any nonlinear dispersive wave equation \eqref{fg-fam}
with given nonlinearities $f(u,u_x)$ and $g(u,u_x)$, 
the dynamics of $A(t)$ and $X(t)$ are described by the coupled nonlinear ODEs 
\begin{equation}\label{1peakon-ODEs}
\dot A = \tfrac{1}{2} [F(u,u_x)]_{X},
\quad
\dot X = -\tfrac{1}{2} [G(u,u_x)]_{X}/A
\end{equation}
which arise directly from the $N=1$ case of the $N$-peakon dynamical system \eqref{fgpeakonsys},
with 
\begin{equation}\label{FG}
F=\int f\,d u_x,
\quad
G=\int g\,d u_x .
\end{equation}
Since $u_x = -\sgn(x-X) A e^{-|x-X|}$ is jump discontinuous at $x=X$, 
we see that the resulting jumps in $F$ and $G$ which occur in the ODEs \eqref{1peakon-ODEs}
are given by 
$[F(u,u_x)]_{X}= F(A,-A)-F(A,A)=-2F^-(A,A)$ 
and $[G(u,u_x)]_{X}= G(A,-A)-G(A,A) = -2G^-(A,A)$,
where 
\begin{equation}\label{FG-decomp}
F^\mp(u,u_x) = \tfrac{1}{2}( F(u,u_x)\mp F(u,-u_x) ), 
\quad
G^\mp(u,u_x) = \tfrac{1}{2}( G(u,u_x)\mp G(u,-u_x) )
\end{equation}
are the odd and even parts of $F$ and $G$ under a reflection $(u,u_x)\to (u,-u_x)$. 
The ODEs can thus be expressed succinctly as
\begin{equation}\label{ODEs}
\dot A = -F^-(A,A) , 
\quad
\dot X = G^-(A,A)/A  .
\end{equation}
Since these ODEs \eqref{ODEs} only involve the odd parts of $F$ and $G$,
we can write them in an equivalent form in terms of the even parts of $f$ and $g$,
by using the relations 
$F^-(A,A) = \int_{-A}^{A} f^+(A,u_x)\, du_x$
and $G^-(A,A) = \int_{-A}^{A} g^+(A,u_x)\, du_x$,
where \label{fg-decomp}
\begin{equation}
f^\pm(u,u_x) = \tfrac{1}{2}( f(u,u_x)\pm f(u,-u_x) ), 
\quad
g^\pm(u,u_x) = \tfrac{1}{2}( g(u,u_x)\pm g(u,-u_x) )
\end{equation}
are the even and odd parts of $f$ and $g$. 
This establishes the following result. 

\begin{prop}\label{prop:dynpeakedwave}
(i)
The amplitude $A(t)$ and position $X(t)$ of dynamical peaked waves \eqref{1peakon}
satisfy the coupled nonlinear ODEs 
\begin{equation}\label{AXsys}
\dot A = -\int_{-A}^{A} f^+(A,y)\, dy, 
\quad
\dot X = (1/A)\int_{-A}^{A} g^+(A,y)\, dy
\end{equation}
in terms of the nonlinearities $f(u,u_x)$ and $g(u,u_x)$ 
in a given nonlinear dispersive wave equation \eqref{fg-fam}. 
(ii) 
A dynamical peaked wave will be a travelling-wave peakon, with a non-zero speed, 
if and only if 
$A=a$ is an arbitrary constant 
and $\dot X=c$ is a non-zero constant. 
These two conditions hold when (and only when) $f(u,u_x)$ and $g(u,u_x)$ satisfy
\begin{equation}\label{FGodd-cond}
\int_{-a}^{a} f^+(a,y)\, dy =0, 
\quad
\int_{-a}^{a} g^+(a,y)\, dy \neq 0, 
\quad
\text{for all }
a\in\Rnum . 
\end{equation}
\end{prop}

We now derive an explicit form of the nonlinearities $f$ and $g$ 
for which dynamical peaked waves are more general than travelling-wave peakons. 
In particular, we find the general solution of the conditions \eqref{FGodd-cond}
for $f$ and $g$. 
This will be accomplished by starting from the decomposition of $F$ and $G$ 
into odd and even parts \eqref{FG-decomp} with respect to the reflection $u_x\to -u_x$. 

First, we express $F^- = u_x \tilde F^+$ and $G^- = u_x \tilde G^+$,
where $\tilde F^+,\tilde G^+$ are reflection-invariant functions of $u,u_x$. 
Next, from relation \eqref{FG}, we have 
$f= F_{u_x} = F^+_{u_x} +\tilde F^+ + u_x\tilde F^+_{u_x}$
and 
$g= G_{u_x} = G^+_{u_x} +\tilde G^+ + u_x\tilde G^+_{u_x}$. 
Matching these expressions to the decomposition of $f$ and $g$ into 
odd and even parts \eqref{fg-decomp}, 
we obtain the relations 
$f^-=F^+_{u_x}$ and $f^+ = \tilde F^+ + u_x\tilde F^+_{u_x}$,
and likewise 
$g^-= G^+_{u_x}$ and $g^+ = \tilde G^+ + u_x\tilde G^+_{u_x}$,
since all of $F^+_{u_x},\tilde F^+_{u_x},G^+_{u_x},\tilde G^+_{u_x}$
are odd in $u_x$. 
Then these relations determine 
$F^+ = \int_{u}^{u_x} f^-\, du_x +F_0(u)$
and 
$G^+ = \int_{u}^{u_x} g^-\, du_x +G_0(u)$, 
as well as
$\tilde F^+ = \int_{u}^{u_x} \tilde f^-\, du_x +f_0(u)$,
and 
$\tilde G^+ = \int_{u}^{u_x} \tilde g^-\, du_x + g_0(u)$,
where we have written 
$\tilde f^-= \tilde F^+_{u_x}$ and $\tilde g^-= \tilde G^+_{u_x}$. 
As a result, 
we obtain 
\begin{align}
f(u,u_x) & =f^-(u,u_x) + \int_{u}^{u_x} \tilde f^-(u,y)\, dy +u_x \tilde f^-(u,u_x)  +f_0(u), 
\label{f-expr}
\\
g(u,u_x) & = g^-(u,u_x) + \int_{u}^{u_x} \tilde g^-(u,y)\, dy +u_x \tilde g^-(u,u_x)  +g_0(u),
\label{g-expr}
\end{align}
and
\begin{align}
F(u,u_x) & =\int_{u}^{u_x} f^-(u,y)\, dy +u_x\int_{u}^{u_x} \tilde f^-(u,y)\, dy +f_0(u)u_x +F_0(u) , 
\label{F-expr}
\\
G(u,u_x) & =\int_{u}^{u_x} g^-(u,y)\, dy +u_x\int_{u}^{u_x} \tilde g^-(u,y)\, dy +g_0(u)u_x +G_0(u) , 
\label{G-expr}
\end{align}
where the functions $f^-,g^-,\tilde f^-,\tilde g^-$ are odd in $u_x$. 
This representation for the functions $f,g,F,G$ is nicely adapted to the form of 
the travelling-wave conditions \eqref{FGodd-cond},
since we have 
$F^-(a,a) = a f_0(a)$ and $G^-(a,a) = a g_0(a)$.
As an immediate consequence, 
the following classification of nonlinearities is obtained. 

\begin{lem}\label{lem:dotA}
A nonlinear dispersive wave equation \eqref{fg-fam}
possesses dynamical peaked waves \eqref{1peakon}
with a time-dependent amplitude $A(t)$ if and only if 
the nonlinearity $f(u,u_x)$ satisfies $\int_{-a}^{a} f^+(a,y)\, dy \neq 0$ for some $a\in \Rnum$. 
This inequality is equivalent to the condition that the term $f_0(u)$ 
in the representation \eqref{f-expr} for $f(u,u_x)$ is non-zero. 
\end{lem}

We next derive a similar necessary and sufficient condition 
for dynamical peaked waves to be accelerating. 
From the ODEs \eqref{AXsys}, 
the acceleration of $X(t)$ is given by 
\begin{equation}\label{accelODE}
\ddot X = -F^-(A,A) \alpha(A) 
\end{equation}
where 
$\alpha(A) = \big(A(g^+(A,A)+ G^-_{u}(A,A)) -G^-(A,A))\big)/A^2$. 
Using the representations \eqref{g-expr} and \eqref{G-expr} for $g$ and $G$, 
we find $g^+(A,A) = A\tilde g^-(A,A) +g_0(A)$ 
and $G^-_{u}(A,A) = A(g_0'(A)-\tilde g^-(A,A))$, 
which yields 
\begin{equation}
\alpha(A) = g_0'(A)
\end{equation}
This expression, combined with the preceding expressions for $\ddot X$, 
leads to the following classification of nonlinearities. 

\begin{lem}\label{lem:dotX}
A nonlinear dispersive wave equation \eqref{fg-fam}
possesses dynamical peaked waves \eqref{1peakon}
with a non-zero acceleration $\ddot X(t)$ if and only if 
the amplitude $A(t)$ is time-dependent
and the nonlinearity $g(u,u_x)$ satisfies 
$\int_{-a}^{a} (g^+(a,y) - a g^+_{u}(a,y))\, dy \neq ag^+(a,a)$ for some $a\in \Rnum$. 
This inequality is equivalent to the condition that the term $g_0(u)$ 
in the representation \eqref{g-expr} for $g(u,u_x)$ is non-constant. 
\end{lem}

Since dynamical peaked waves are not classical solutions, 
we remark that some mild regularity conditions on the nonlinearities $f$ and $g$ 
are additionally needed in Proposition~\ref{prop:dynpeakedwave} and Lemmas~\ref{lem:dotA} and~\ref{lem:dotX}
so that a nonlinear dispersive wave equation \eqref{fg-fam} 
possesses a weak formulation.
In particular $f,g$ being $C^0$ in $u$ and $C^1$ in $u_x$ is sufficient. 

For doing analysis of a nonlinear dispersive wave equation \eqref{fg-fam}, 
such as proving local well-posedness and global existence, 
it can be useful to consider $f$ and $g$ to have stronger regularity. 
In particular, 
if $f$ and $g$ are, locally, analytic functions of $u$ and $u_x$, 
then the representations \eqref{f-expr} and \eqref{g-expr} of $f$ and $g$ 
become more complicated, as follows. 
First, 
we define the functions 
$f_1=f^-/u_x$, $f_2=\tilde f^-/u_x$, $g_1=g^-/u_x$, $g_2=\tilde g^-/u_x$, 
which are even in $u_x$. 
Next, we change variables from $(u,u_x)$ to $(v_+,v_-)=(u+u_x,u-u_x)$, 
whereby a reflection $u_x\to -u_x$ becomes a permutation $v_\pm\to v_\mp$. 
A useful result from invariant theory \cite{Olv-book2} is that 
any permutation-invariant function of $v_\pm$ that is also analytic 
can be expressed as a function of $v_++v_-=2u$ and $v_+v_-=u^2-u_x^2$. 
Hence we can take $f_1^+,f_2^+,g_1^+,g_2^+$ to be functions of $u$ and $u^2-u_x^2$. 
Then, the representations \eqref{f-expr} and \eqref{g-expr} are given by 
\begin{equation}
\begin{aligned}
f(u,u_x) & = f_0(u)+u_x f_2(u,u^2-u_x^2)  +u_x^2 f_1(u,u^2-u_x^2)  -\tfrac{1}{2}\int_{0}^{u^2-u_x^2} f_1(u,y)\, dy , 
\\
g(u,u_x) & = g_0(u)+u_x g_2(u,u^2-u_x^2)  +u_x^2 g_1(u,u^2-u_x^2)  -\tfrac{1}{2}\int_{0}^{u^2-u_x^2} g_1(u,y)\, dy ,
\end{aligned} 
\end{equation}
where $f_1,f_2,g_1,g_2$ are analytic functions of $u$ and $u^2-u^2_x$,
and $f_0,g_0$ are analytic functions of $u$. 

\subsection{Momentum and Sobolev norm}

The conditions stated in Lemmas~\ref{lem:dotA} and~\ref{lem:dotX}
have a connection to the momentum and the Sobolev norm
for solutions $u(x,t)$ of nonlinear dispersive wave equations \eqref{fg-fam}. 

The momentum and Sobolev ($H^1)$ norm of a solution $u(x,t)$ 
are given by the respective integrals
\begin{equation}\label{mom-H1norm}
M =\int_\Rnum m\,dx,
\quad
||u||_{H^1} = \int_\Rnum u^2 + u_x^2\,dx.
\end{equation}
When $u(x,t)$ has sufficiently rapid decay as $|x|\to\infty$, 
these integrals can be expressed as 
$M =\int_\Rnum u\,dx$
and $||u||_{H^1} = \int_\Rnum um\,dx$. 
In particular, $|u|=O(1)$ and $u_x=o(1)$ as $|x|\to\infty$ suffices. 

For dynamical peaked wave solutions \eqref{1peakon},
we can easily evaluate the integrals \eqref{mom-H1norm}, 
which yields 
\begin{equation} 
M = 2A(t), 
\quad
||u||_{H^1} = 2A(t)^2.
\end{equation}
Hence, 
the momentum and the Sobolev norm are conserved if and only if 
the amplitude $A(t)$ is constant. 
From Lemma~\ref{lem:dotX},
this implies that accelerating dynamical peaked waves do not conserve
both the momentum and the Sobolev norm. 

Consequently, the momentum and the Sobolev norm cannot be conservation laws 
holding for nonlinear dispersive wave equations \eqref{fg-fam} that possess
accelerating dynamical peaked waves. 
In \Ref{AncRec}, 
necessary and sufficient conditions on the nonlinearities $f$ and $g$ 
have been obtained for these two conservation laws to hold. 
The conditions consist of 
$f^-= u_x h_1(u^2-u_x^2)$, $f^+=h_0u/(u^2-u_x^2)$, 
for conservation of momentum; 
and 
$f^-=u_x (h^+(u,u_x) +h_1(u^2-u_x^2)/u)$, 
$f^+=u_x h^-(u,u_x) +h_0/(u^2-u_x^2)$, 
$g^-=uh^-(u,u_x) +h_0u/(u_x(u^2-u_x^2))$, 
$g^+=uh^+(u,u_x) -h_1(u^2-u_x^2)$, 
for conservation of the Sobolev norm. 
Here $h_0$ is an arbitrary constant, $h_1$ is an arbitrary function of $u^2-u_x^2$,
and $h$ is an arbitrary function of $u,u_x$. 
Therefore, 
a nonlinear dispersive wave equation \eqref{fg-fam} will possess
accelerating dynamical peaked waves if and only if the nonlinearities $f$ and $g$ 
do not have the above forms.

\section{Main results}
\label{results}

The coupled nonlinear ODEs \eqref{AXsys} governing dynamical peaked waves \eqref{1peakon}
comprise a separable system that has a straightforward quadrature
for the amplitude $A(t)$ and position $X(t)$ of the waves. 
We will look specifically at the situation 
when $\dot A(t)\neq 0$ for some time interval $t\in [t_1,t_2)\subseteq\Rnum$. 
Then the amplitude ODE yields 
\begin{equation}\label{A-integral}
\int_{A_1}^{A} \frac{dy}{F^-(y,y)} = t_1 -t
\end{equation}
which determines $A(t)$ by quadrature, 
while the position ODE gives
\begin{equation}\label{X-integral}
X-X_1 = \int_{t_1}^{t}\frac{G^-(A,A)}{A}\, dt 
= -\int_{A_1}^{A}\frac{G^-(y,y)}{yF^-(y,y)}\, dy 
\end{equation}
which determines $X(t)$ in terms of $A(t)$. 
Here we assume $F^-(A,A)\neq 0$ holds at the initial time $t=t_1$. 

This solution \eqref{A-integral}--\eqref{X-integral} can be written in another way, 
which brings out the role of acceleration. 
We begin by defining 
\begin{equation}\label{c-tau}
c(A) = \frac{G^-(A,A)}{A},
\quad
\tau(A) = -\int \frac{dA}{F^-(A,A)},
\end{equation}
which physically represent a speed function and a time function. 
Note $\tau(A)$ has an inverse at least locally near $A=A_1$. 
We next observe 
\begin{equation}\label{alpha}
c'(A) = \big( A(g^+(A,A) + G^-_{u}(A,A)) -G^-(A,A) \big)/A^2 = \alpha(A) 
\end{equation}
from the acceleration equation \eqref{accelODE}. 

\begin{prop}\label{prop:soln}
In terms of the functions \eqref{c-tau}--\eqref{alpha}, 
the amplitude and position of dynamical peaked waves \eqref{1peakon} 
are given by 
\begin{align}
A(t) & = \tau^{-1}(t-t_1+\tau_1),
\label{Asoln}
\\
X(t) & = X_1 +\int_{A_1}^{A(t)} \tau'(y) c(y)\, dy
= X_1 -c_1 \tau_1  + (t-t_1+\tau_1) c(A(t)) -\int_{A_1}^{A(t)} \tau(y)\alpha(y)\, dy,
\label{Xsoln}
\end{align}
where $\tau_1=\tau(A_1)$ and $c_1=c(A_1)$. 
These two expressions \eqref{Asoln}--\eqref{Xsoln}, 
or equivalently the integral expressions 
\begin{equation}\label{AX-integral}
\int_{A}^{A_1} \frac{dy}{yf_0(y)} = t-t_1 , 
\quad 
\int_{A}^{A_1}\frac{g_0(y)}{yf_0(y)}\, dy = X-X_1, 
\end{equation}
completely determine the time evolution of all dynamical peaked waves. 
\end{prop}

Several interesting consequences can now be inferred about the behaviour of dynamical peaked waves. 

Firstly, 
the behaviour of dynamical peaked waves depends only on the function $\tau$, and the functions $c$ or $\alpha$,
which are determined by the form of the nonlinearities $f$ and $g$ 
in a given nonlinear dispersive wave equation.
In particular, 
from expressions \eqref{c-tau} 
we have $F^-(u,u)= -1/\tau'(u)$ and $G^-(u,u) = u c(u) = u \int \alpha(u)\, du$, 
which are directly related to $f$ and $g$ through their the representations \eqref{f-expr}--\eqref{g-expr} 
given in terms of two functions $f_0(u),g_0(u)$, along with four functions of $u$ and $u_x$.
From these representations, 
along with the corresponding representations \eqref{F-expr}--\eqref{G-expr} of $F$ and $G$,
we see that 
\begin{equation}\label{f0-F-g0-G}
f_0(u)= F^-(u,u)/u,
\quad
g_0(u)= G^-(u,u)/u. 
\end{equation}
Consequently, we obtain the relations 
\begin{equation}
f_0(u) = -1/(u\tau'(u)),
\quad
g_0(u) = c(u) = \txtint \alpha(u)\, du.
\end{equation}
Then these two relations can be inverted to get
\begin{equation}\label{rels}
\tau(y) = -\int \frac{dy}{yf_0(y)},
\quad
c(y) = g_0(y), 
\quad
\alpha(y) = g_0'(y).
\end{equation}
Since $f$ and $g$ each depend on two functions of $u$ and $u_x$, 
in addition to the functions $f_0$ and $g_0$, 
we conclude that there is a large class of nonlinearities yielding the same functions $\tau,c,\alpha$. 

Secondly, 
there is no restriction on the time function $\tau(y)$ 
and either the speed or acceleration functions $c(y)$, $\alpha(y)$, 
since $f_0$ and $g_0$ are determined if $\tau$ and $c$ or $\alpha$ are specified. 
This means that any behaviour for the amplitude $A(t)$ and the position $X(t)$ 
can be selected by an appropriate choice of the functions $f_0(u)$ and $g_0(u)$ 
appearing in the nonlinearities $f(u,u_x)$ and $g(u,u_x)$.
In particular, 
we can straightforwardly characterize which nonlinear dispersive wave equations \eqref{fg-fam} 
will possess dynamical peaked waves having any specified behaviour for $A(t)$ and $X(t)$
by analysis of the ODEs \eqref{ODEs} and \eqref{accelODE}
expressed in terms of the functions \eqref{rels}:
\begin{equation}\label{ODEs-fg}
\dot A = -A f_0(A) , 
\quad
\dot X = g_0(A)  ,
\quad
\ddot X = -A f_0(A) \alpha(A) . 
\end{equation}

To illustrate this connection, 
we now consider the following different types of behaviour for the amplitude $A(t)$.  
Note that local existence of $A(t)$ holds whenever the function $1/(uf_0(u))$ is locally integrable with respect to $u$. 

\begin{thm}\label{thm:A-behav}
Suppose $A(t)$ exists on some time interval $t_0\leq t<t^*$. \hfil
\begin{enumerate}[(i)]
\item
Finite asymptote: 
$A(t)\to A_\infty=\const(\neq 0)$ as $t\to \infty$
iff $f_0(u)=O(u-A_\infty)$ as $u\to A_\infty$. 
\item
Unbounded:
$|A|\to \infty$ as $t\to \infty$ 
iff $f_0(u)=O(1)$ as $u\to\infty$. 
\item
Blow-up:
$|A|\to \infty$ as $t\to t^* <\infty$ 
iff $1/f_0(u)=o(1)$ as $u\to\infty$. 
\item 
Extinction:
$A,\dot A\to 0$ as $t\to t^*\leq \infty$
iff $1/f_0(u)=o(1)$ and $uf_0(u)=o(1)$ as $u\to 0$ when $t^*<\infty$, 
or $f_0(u)=O(1)$ as $u\to 0$ when $t^*=\infty$.
\end{enumerate}
\end{thm}

\begin{proof}
We will use the ODEs \eqref{ODEs-fg} and the quadrature \eqref{Asoln}. 
Part (i) is equivalent to the condition on the function $\tau'(u)= -1/(uf_0(u))$ 
for the quadrature $\lim_{A\to A_\infty}\tau(u)|^{A}_{A_0}$ to diverge to $\infty$, 
while part (ii) is similarly equivalent to the condition 
for $\lim_{A\to \infty}\tau(u)|^{A}_{A_0}$ to diverge to $\infty$. 
Part (iii) is equivalent to the condition for the quadrature 
$\lim_{A\to \infty}\tau(u)|^{A}_{A_0}$ to converge to a finite value $t^*-t_0$. 
Finally, part (iv) is equivalent to the condition 
for the function $\tau'(u)= -1/(uf_0(u))$ to go to zero and for its quadrature 
$\lim_{A\to 0}\tau(u)|^{A}_{A_0}$ to either converge to a finite value $t^*-t_0$ or diverge to $\infty$. 
\end{proof}

We next consider different types of behaviour for the position $X(t)$. 
Note that local existence of $X(t)$ holds whenever $A(t)$ exists 
and the function $g_0(u)/(uf_0(u))$ is locally integrable with respect to $u$. 
The following Theorem is an immediate consequence of the ODEs \eqref{ODEs-fg} for $\dot X$ and $\ddot X$,
as well as use of the quadrature \eqref{Xsoln} for $X$ 
similarly to the proof of Theorem~\ref{thm:A-behav}. 

\begin{thm} \label{thm:X-behav}
Suppose $A(t),X(t)$ exist on some time interval $t_0\leq t< t^*$ 
such that $A_*=\lim_{t\to t^*}A(t)$ exists (including when $|A_*|=\infty$). \hfil
\begin{enumerate}[(i)]
\item
Finite asymptotic speed:
$\ddot X(t)\to 0$ and $\dot X(t)\to c_\infty=\const$ as $t\to \infty$
iff $g_0(u)\to c_\infty$ and $uf_0(u)g_0'(u)=o(1)$ as $u\to A_\infty$. 
\item 
Runaway:
$|X(t)|,|\dot X(t)|\to \infty$ as $t\to \infty$ 
iff $1/g_0(u)=o(1)$ as $u\to A_\infty$
and either $f_0(u)=O(u-A_\infty)$ ($f_0(u)/g_0(u)=o(u-A_\infty)$) as $u\to A_\infty$ when $0<|A_\infty|<\infty$,
or $f_0(u)=O(1)$ ($f_0(u)/g_0(u)=o(1)$) as $u\to A_\infty$ when $|A_\infty|=0,\infty$. 
\item 
Braking:
$\dot X(t),\ddot X(t)\to 0$ as $t\to t^*\leq\infty$ 
iff $g_0(u)=o(1)$ and $uf_0(u)g_0'(u)=o(1)$ as $u\to A_*$. 
\item
Wheelspin-limit:
$|X(t)|$ bounded and $|\dot X(t)|\to \infty$ as $t\to t^*<\infty$ 
iff $1/g_0(u) = o(1)$ as $u\to A_*$ 
and either $g_0(u)(u-A_*)/f_0(u)=o(1)$, $(u-A_*)/f_0(u)=o(1)$ when $0<|A_*|<\infty$, 
or $1/f_0(u)=o(1)$, $g_0(u)/f_0(u)=o(1)$ when $|A_*|=0,\infty$. 
\item
Finite-time runaway:
$|X(t)|,|\dot X(t)|\to \infty$ as $t\to t^*<\infty$ 
iff $1/g_0(u)=o(1)$ as $u\to A_*$
and either $(u-A_*)/f_0(u)=o(1)$, $f_0(u)/g_0(u)=o(u-A_*)$ when $0<|A_*|<\infty$,
or $1/f_0(u)=o(1)$, $f_0(u)/g_0(u)=o(1)$ when $|A_*|=0,\infty$. 
\end{enumerate}
\end{thm}

Finally, combining part (i) of both Theorems, 
we obtain necessary and sufficient conditions
on the nonlinearities $f(u,u_x)$ and $g(u,u_x)$ 
so that the asymptotic behaviour of dynamical peaked waves 
is a travelling-wave peakon. 

\begin{cor}\label{cor:travellingwave-behav}
Suppose $A(t),X(t)$ exist for all $t_0\leq t < \infty$. 
Then $u(x,t)= A(t) e^{-|x-X(t)|} \to a e^{-|x-c t -x_0|}$ as $t\to \infty$,
where $a,c,x_0=\const$, 
iff $g_0(u)\to c$, $f_0(u)=O(u-a)$, and $f_0(u)g_0'(u)=o(u-a)$ as $u\to a$. 
\end{cor}

\subsection{Time-evolution of accelerating dynamical peakons}
\label{accel-peakons}

To illustrate Theorems~\ref{thm:A-behav} and~\ref{thm:X-behav}, 
we now consider a family of wave equations 
\begin{equation}\label{ex-fam}
m_t + \kappa u^p m +\lambda (u^q m)_x=0,
\quad
p,q\neq 0
\end{equation}
in which the nonlinearities involve only $u$,
where $\kappa,\lambda$ are non-zero constants. 
The time evolution of the dynamical peaked wave solutions of these wave equations \eqref{ex-fam} 
turn outs to exhibit a wide variety of behaviours. 

For this family \eqref{ex-fam}, 
we have $f=\kappa u^p =f^+$ and $g=\lambda u^q =g^+$,
and thus $F= \kappa u^p u_x=F^-$ and $G=\lambda u^q u_x=G^-$,
so then relation \eqref{f0-F-g0-G} yields $f_0 = \kappa u^p$ and $g_0= \lambda u^q$. 
From Lemma~\ref{lem:dotA} combined with Proposition~\ref{prop:soln}, 
we see that the dynamical peaked wave solutions \eqref{1peakon} 
have a time-dependent amplitude given by 
\begin{equation}\label{A:ex-fam}
A=(p\kappa (t- t_0))^{-1/p},
\quad
t_0=\const
\end{equation}
with initial value $A(0)= (-p\kappa t_0)^{-1/p}$. 
Then, since $g_0$ is non-zero, 
we conclude from Lemma~\ref{lem:dotX} 
that the acceleration of the solutions is non-zero. 

From the evolution ODEs \eqref{ODEs-fg}, 
we obtain 
$\dot X = \lambda A^q$ 
and 
$\ddot X = -q\kappa A^{p+q}$,
yielding the speed and the acceleration 
\begin{equation}
\dot X =\lambda (p\kappa (t-t_0))^{-q/p},
\quad
\ddot X = -q\kappa\lambda (p\kappa (t-t_0))^{-(1+q/p)}.
\end{equation}
Hence, the position of the peakons is given by 
\begin{equation}\label{X:ex-fam}
X = 
\begin{cases} 
X_0 + \frac{\lambda}{(p-q)\kappa} (p \kappa (t-t_0))^{1-q/p}, & q\neq p,
\\
X_0 + \frac{\lambda}{p\kappa} \ln(|t-t_0|),  & q=p,
\end{cases}
\quad
X_0=\const
\end{equation}
with initial value 
$X(0) = X_0 + \frac{\lambda}{(p-q)\kappa} (-p \kappa t_0)^{1-q/p}$ when $q\neq p$, 
and $X(0)=X_0 + \frac{\lambda}{p\kappa} \ln(|t_0|)$ when $q=p$. 

The type of behaviour exhibited by these accelerating dynamical peakons \eqref{A:ex-fam}--\eqref{X:ex-fam}
for $t\geq 0$ 
depends on the nonlinearity powers $p,q$, and the signs of $\kappa,\lambda$.
We will proceed by taking $\lambda>0$,
which corresponds to the speed $\dot X$ being non-negative. 

First, if $p$ and $\kappa$ are positive, 
then $A,X$ exist for all $t\geq 0$, 
with $t_0$ being chosen to be negative. 
As $t\to t^*=\infty$, we have $A\to 0$ and $\dot A\to 0$,
whereby $A$ goes extinct. 
This behaviour is in accordance with part (iv) of Theorem~\ref{thm:A-behav}, 
since $f_0(u)=\kappa u^p =o(1)$ and $uf_0(u)=\kappa u^{p+1} =o(1)$ as $u\to 0$. 

If $q>p$,
then $X\to X_0$ while $\dot X,\ddot X\to 0$. 
Hence, $X$ exhibits braking behaviour. 
This in accordance with part (iii) of Theorem~\ref{thm:X-behav},
since $g_0(u)=\lambda u^q=o(1)$ and $uf_0(u)g_0'(u)=\lambda\kappa q u^{p+q}=o(1)$ 
as $u\to 0$. 

If $q\leq p$, 
then $X\to \infty$. 
Additionally, when $q>0$, we have $\dot X,\ddot X\to 0$,
and thus $X$ exhibits asymptotic braking. 
In contrast, when $q<0$, we have $\dot X\to \infty$,
whereby $X$ exhibits runaway behaviour. 
These two behaviours are in accordance with parts (iii) and (v) of Theorem~\ref{thm:X-behav},
since as $u\to 0$, 
$g_0(u)=\lambda u^q =o(1)$ and $uf_0(u)g_0'(u)= q\kappa\lambda u^{p+q}=o(1)$
in the first case, 
and $1/g_0(u)=\lambda^{-1} u^{|q|} =o(1)$ and $f_0(u)= \kappa u^p=o(1)$
in the second case. 

Second, if $p$ is positive but $\kappa$ is negative, 
then $A,X$ exist for $0\leq t < t_0$, provided $t_0$ is chosen to be positive. 
As $t\to t^*=t_0$, we have $A\to \infty$, which is a blow-up. 
This behaviour is in accordance with part (iii) of Theorem~\ref{thm:A-behav},
since $1/f_0(u)=\kappa^{-1} u^{-p} =o(1)$ as $u\to \infty$. 

If $q\geq p$, 
then $X,\dot X\to \infty$. 
This finite-time runaway behaviour is in accordance with part (v) of Theorem~\ref{thm:X-behav},
since as $u\to \infty$, 
$1/g_0(u)=\lambda^{-1}u^{-q} =o(1)$, $1/f_0(u)=\kappa^{-1}u^{-p} =o(1)$, 
and $f_0(u)/g_0(u)=\kappa\lambda^{-1} u^{p-q} =o(1)$. 

If $q<p$, 
then $X\to X_0$. 
Additionally, when $q>0$, 
we have $\dot X\to \infty$,
which describes a wheelspin-limit. 
In contrast, when $q<-p$,
we have $\dot X,\ddot X\to 0$,
whereby $X$ exhibits braking. 
These two behaviours are in accordance with parts (iii) and (v) of Theorem~\ref{thm:X-behav}, since as $u\to \infty$, 
$1/g_0(u)=\lambda^{-1} u^{-q}=o(1)$, $1/f_0(u)=\kappa^{-1} u^{-p}=o(1)$
and $g_0(u)/f_0(u)=\lambda\kappa^{-1} u^{-(p-q)}=o(1)$
in the first case,
and 
$g_0(u)=\lambda u^{-|q|}=o(1)$ and $uf_0(u)g_0'(u)=q\kappa\lambda u^{-(|q|-p)}=o(1)$
in the second case. 
Finally, when $0<q<-p$,
we have $\dot X\to 0$ while $\lim_{t\to t^*}\ddot X\neq 0$. 
This is like a thrust-reverse braking behaviour. 

Third, if $p$ and $\kappa$ are negative,
then $A,X$ exist for all $t\geq 0$, when $t_0$ is chosen to be non-positive. 
As $t\to t^*=\infty$, we have $A\to \infty$. 
This unbounded behaviour is in accordance with part (ii) of Theorem~\ref{thm:A-behav}, 
since $f_0(u)=\kappa u^{-|p|} =o(1)$ as $u\to\infty$. 

If $q\geq p$,
then $X\to \infty$. 
When $q>0$, we also have $\dot X\to\infty$,
and hence $X$ exhibits runaway behaviour. 
This in accordance with part (ii) of Theorem~\ref{thm:X-behav},
since as $u\to \infty$, 
$1/g_0(u)=\lambda^{-1}u^{-q} =o(1)$ 
and $f_0(u)=\kappa u^{-|p|} =o(1)$. 
Instead, when $0>q\geq p$, 
we have $\dot X,\ddot X\to 0$. 
This describes asymptotic braking,
which is in accordance with part (iii) of Theorem~\ref{thm:X-behav},
since $g_0(u)=\lambda u^{-|q|} =o(1)$ and $uf_0(u)g_0'(u)=q\kappa\lambda u^{-|p|-|q|} =o(1)$
as $u\to \infty$. 

Last, if $p$ is negative but $\kappa$ is positive, 
then $A,X$ exist for $0\leq t < t_0$, with $t_0$ chosen to be positive. 
As $t\to t^*=t_0$, 
we have $A\to 0$. 
When $-p<1$ , we also have $\dot A\to 0$, 
and otherwise we have $\lim_{t\to t^*}\dot A\neq 0$ when $-p\geq 1$.
Thus, in the first case, $A$ exhibits a finite-time extinction,
which is in accordance with part (iv) of Theorem~\ref{thm:A-behav}
since $1/f_0(u)=\kappa^{-1} u^{|p|} =o(1)$ and $uf_0(u)=\kappa u^{1-|p|}=o(1)$ 
as $u\to 0$.  
In the second case, $A$ exhibits a singular behaviour. 

If $q\leq p$, 
then $X,\dot X\to \infty$ as $t\to t_0$, 
and hence $X$ exhibits a finite-time runaway. 
This behaviour is in accordance with part (v) of Theorem~\ref{thm:X-behav},
since $1/f_0(u)=\kappa^{-1}u^{|p|} =o(1)$ and $f_0(u)/g_0(u)=\kappa\lambda^{-1}u^{|q|-|p|} =o(1)$ 
as $u\to 0$. 

If $q>p$,
then $X\to X_0$. 
When, additionally, $q<0$,
we have $\dot X\to \infty$,
which describes a wheelspin-limit. 
This behaviour in accordance with part (iv) of Theorem~\ref{thm:X-behav},
since as $u\to 0$, 
$1/f_0(u)=\kappa^{-1}u^{|p|} =o(1)$ and $1/g_0(u)=\lambda^{-1}u^{|q|} =o(1)$. 
When $q> -p$,
we have $\dot X,\ddot X\to 0$. 
Hence, $X$ exhibits braking,
which is in accordance with part (iii) of Theorem~\ref{thm:X-behav},
since $g_0(u)=\lambda u^{q} =o(1)$ and $uf_0(u)g_0'(u)=q\kappa\lambda u^{q-|p|} =o(1)$
as $u\to 0$. 
Finally, when $0<q<-p$,
we instead have $\dot X\to 0$ while $\lim_{t\to t^*}\ddot X\neq 0$,
which is like a thrust-reverse braking behaviour. 

\subsection{Stationary peakons with time-evolving amplitude}
\label{stationary-peakons}

We can completely alter the speed and acceleration of 
the previous dynamical peaked waves \eqref{A:ex-fam}--\eqref{X:ex-fam}
by changing the nonlinearity in the $m_x$ term in the wave equations \eqref{ex-fam} 
to include $u_x$:
\begin{equation}\label{ex-fam-ux}
m_t + \kappa u^p m +\lambda(u^q u_x m)_x=0,
\quad
p\neq 0.
\end{equation}
Here we have $g=\lambda u^q u_x=g^-$, while $f=\kappa u^p =f^+$ as before. 
This gives $F= \kappa u^p u_x=F^-$ and $G=\tfrac{1}{2}\lambda u^q u_x^2=G^+$,
and consequently we get $g_0= 0$ while $f_0 = \kappa u^p$ 
from relation \eqref{f0-F-g0-G}. 

The resulting dynamical peaked wave solutions \eqref{1peakon}, 
given by Proposition~\ref{prop:soln}, 
thus have the same time-dependent amplitude \eqref{A:ex-fam}
as the solutions for the previous family \eqref{ex-fam}, 
but their speed is now zero due to $\dot X =g_0=0$
from the evolution ODEs \eqref{ODEs-fg}. 

Hence, these peakons are stationary while the time-evolution of their amplitude exhibits
finite-time extinction or blow up, and long-time extinction or unbounded behaviour. 

This behaviour remains the same if the nonlinearity in the $m_x$ term is changed in other ways, such as in the family of wave equations
\begin{equation}\label{ex-fam2-ux}
m_t + \kappa u^p m +\lambda (u^q (u^2-3u_x^2) m)_x=0,
\quad
p\neq 0.
\end{equation}
We now have $g=\lambda u^q( u^2-3u_x^2)=g^-$,
which gives $G=\lambda u^q (u^2 - u_x^2)u_x=G^+$,
and so we again get $g_0= 0$ from relation \eqref{f0-F-g0-G}. 
This yields the same stationary peakon solutions as obtained for the previous family of wave equations \eqref{ex-fam-ux}. 

We also remark that these peakons can be made to move with constant speed $c\neq0$ 
if the constant $c$ is added to the $m_x$ term.

\section{Examples}
\label{examples}

Five examples of different types of interesting peaked dynamical waves 
will now be discussed. 
In each example, 
the amplitude, speed, and position of these waves are obtained in an explicit form,
and their asymptotic behaviour is described.

\subsection{Travelling wave peakons}
\label{trav-wave-peakons}

We begin by considering the family of nonlinear dispersive wave equations 
\begin{equation}\label{ex1}
m_t + \kappa u^p u_x m +((u^{p-1}(u^2+\lambda u_x^2))m)_x=0
\end{equation}
where $p$ is a nonlinearity power, and $\kappa,\lambda$ are constants. 
This family is a nonlinear generalization of 
all of the known integrable peakon equations ---
CH ($p=0$, $\lambda=0$, $\kappa=1$), 
DP ($p=0$, $\lambda=0$, $\kappa=2$), 
N ($p=1$, $\lambda=0$, $\kappa=1$), 
mCH ($p=1$, $\lambda=-1$, $\kappa=0$) 
--- as well as the $b$-family equation ($p=0$, $\lambda=0$, $\kappa=b-1$) 
which unifies the CH and DP equations ($b=2,3$) but otherwise is non-integrable. 
For these equations, dynamical peaked wave solutions are travelling-wave peakons. 

We will now show that the dynamical peaked wave solutions \eqref{1peakon} 
for every wave equation in the family \eqref{ex1} consist of travelling-wave peakons
$u(x,t)= a e^{-|x-c t -x_0|}$. 

The nonlinearities in these equations \eqref{ex1} are given by 
$f=\kappa u^p u_x =f^-$ and $g=u^{p-1}(u^2+\lambda u_x^2)=g^+$. 
Thus, we have 
$F= \tfrac{1}{2}\kappa u^p u_x^2=F^+$ 
and $G=u^{p-1}(u^2+\tfrac{1}{3}\lambda u_x^2) u_x=G^-$, 
and hence 
$f_0 = 0$ and $g_0= (1+\tfrac{1}{3}\lambda)u^{p+1}$ 
from relation \eqref{f0-F-g0-G}. 

Combining Lemmas~\ref{lem:dotA} and~\ref{lem:dotX}, 
we conclude that all dynamical peaked wave solutions \eqref{1peakon} 
are travelling-wave peakons with an arbitrary constant amplitude $a$ 
and a constant speed $c= (1+\tfrac{1}{3}\lambda)a^{p+1}$. 
Notice that $\kappa$ plays no role in the form and the behaviour of these peakons,
similarly to what occurs for the peakons in the $b$-family. 
Moreover, the peakon speed is non-zero, except in the special case where $\lambda=-3$.

\subsection{Asymptotic travelling-wave peakons}
\label{asympt-trav-wave-peakons}

As shown by Corollary~\ref{cor:travellingwave-behav},
many nonlinear dispersive wave equations in the general family \eqref{fg-fam}
possess dynamical peaked waves that asymptotically behave like travelling-wave peakons
\begin{equation}\label{asympt-peakon}
u(x,t) \sim a_\pm e^{-|x-c_\pm t -x_\pm|}
\text{ as } t\to\pm\infty
\end{equation}
where $a_\pm,c_\pm,x_\pm$ are constants. 

We consider, firstly, a specific family of cubic nonlinear wave equations
\begin{equation}\label{ex2}
m_t + \kappa(u-2)(u-1) m +\lambda (u m)_x=0
\end{equation}
where $\kappa,\lambda$ are non-zero constants. 
Since we have $f=\kappa(u-2)(u-1)=f^+$ and $g=\lambda u=g^+$,
this yields $F= \kappa(u-2)(u-1) u_x=F^-$ and $G=\lambda u u_x=G^-$,
and so we get $f_0=f$ and $g_0=g$ from relation \eqref{f0-F-g0-G}. 

We are interested in dynamical peaked wave solutions \eqref{1peakon} that are smooth for all $t$. 
Applying Proposition~\ref{prop:soln}, 
we obtain 
\begin{equation}\label{ex2-AX}
\begin{aligned}
& A=1 + 1/\sqrt{1+e^{2\kappa (t_0- t)}},
\quad
X= 2\lambda (t-t_0) + (\lambda/\kappa)\ln\big(1+\sqrt{1+e^{2\kappa (t_0 -t)}}\big) +X_0,
\\& 
t_0,X_0=\const.
\end{aligned}
\end{equation}
As $t\to -\infty$, the amplitude and position have the asymptotic behaviour
\begin{equation}
A \sim 1  + O(e^{-\kappa |t|}) , 
\quad
X -\lambda t \sim X_0 -\lambda t_0 + O(e^{-\kappa |t|}) , 
\end{equation}
which describes a travelling-wave peakon \eqref{asympt-peakon},
with amplitude $a_{-}=1$, speed $c_{-}=\lambda$, 
and position shift $x_{-}=X_0 - c_{-}t_0$. 
Similarly, as $t\to \infty$, 
the asymptotic behaviour of the amplitude and position are given by 
\begin{equation}
A \sim 2  + O(e^{-2\kappa t}) , 
\quad
X -2\lambda t \sim X_0-2\lambda t_0 +(\lambda/\kappa)\ln(2)+ O(e^{-2\kappa t}), 
\end{equation}
which again describes a travelling-wave peakon \eqref{asympt-peakon},
but with a different amplitude $a_{+}=2$, speed $c_{+}=2\lambda$, 
and position shift $x_{+}=X_0 -c_{+}t_0 +\frac{\lambda}{\kappa}\ln(2)$. 

Therefore, this solution \eqref{ex2-AX} is a dynamical accelerating peakon 
that evolves from a travelling-wave peakon in the asymptotic past 
to a different travelling-wave peakon in the asymptotic future. 
In particular, 
the asymptotic amplitude and speed of the peakon change by 
$\Delta a = a_{+}-a_{-} = 1$ and $\Delta c = c_{+}-c_{-} = \lambda$, 
while the position shifts asymptotically by 
$\Delta x_0 = x_{+}-x_{-} = -\Delta c\, t_0 +\frac{\lambda}{\kappa}\ln(2)
= \lambda(\kappa^{-1}\ln(2) -t_0)$. 
Note the asymptotic change in speed can be obtained more directly 
from the speed-amplitude relation 
\begin{equation}\label{ex2-speed-rel}
\dot X = g_0(A) = \lambda A
\end{equation}
which follows from the time-evolution ODEs \eqref{ODEs-fg}. 
Moreover, note this relation \eqref{ex2-speed-rel}
shows that the direction of the peakon is determined by the sign of $\lambda$:
the peakon moves in the $\pm x$ direction when $\sgn(\lambda)=\pm 1$. 
See Figure~\ref{fig:asymptotic-travelling-wave-peakon}. 

\begin{figure}[h]
\includegraphics[width=0.6\textwidth]{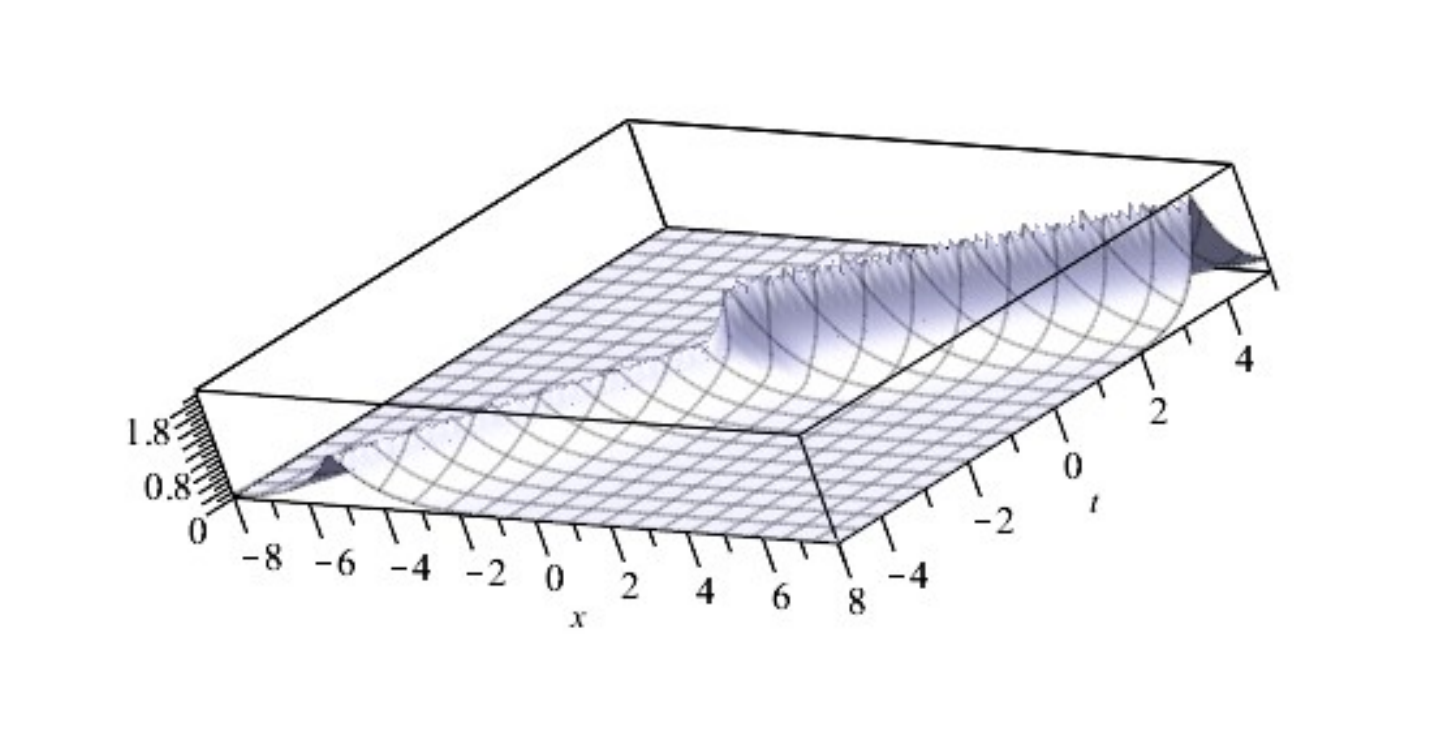}
\caption{Asymptotic travelling-wave peakon}
\label{fig:asymptotic-travelling-wave-peakon}
\end{figure}

The same kind of asymptotic behaviour can be found in more general wave equations:
\begin{equation}\label{ex2a}
m_t + \kappa(u-a_1)(u-a_2)u^p m +\lambda (u^q m)_x=0,
\quad
a_1\neq a_2,
\quad
a_1a_2>0,
\quad
p,q>0.
\end{equation}

\subsection{Direction-reversing peakons}
\label{dir-reverse-peakons}

In the previous family \eqref{ex2} of cubic wave equations,
the dynamical peaked waves \eqref{ex2-AX} do not change direction, 
due to their speed-amplitude relation \eqref{ex2-speed-rel}. 
We will now change the nonlinearity in the $m_x$ term in those wave equations 
so that the speed-amplitude relation dynamically changes sign,
which causes the dynamical peaked waves to reverse direction at some finite time. 

We consider the modified family of cubic nonlinear wave equations
\begin{equation}\label{ex3}
m_t + \kappa(u-2)(u-1) m +\lambda ((3-2u) m)_x=0
\end{equation}
where $\kappa,\lambda$ are non-zero constants. 
For this family, we have $f=\kappa(u-2)(u-1)=f^+$ as before, 
while now $g=\lambda (3-2u)=g^+$. 
Hence, we get $G=\lambda (3-2u) u_x=G^-$ and $F= \kappa(u-2)(u-1) u_x=F^-$,
which yields $g_0=g$ and $f_0=f$ from relation \eqref{f0-F-g0-G}. 

We are again interested in dynamical peaked wave solutions \eqref{1peakon} that are smooth for all $t$. 
Since $f_0$ is unchanged, 
the amplitude of these solutions is also unchanged:
\begin{equation}\label{ex3-A}
A=1 + 1/\sqrt{1+e^{2\kappa (t_0- t)}},
\quad
t_0 =\const
\end{equation}
which has the asymptotic behaviour
\begin{equation}\label{ex3-asympt-A}
A \sim 
\begin{cases}
1  + O(e^{-\kappa |t|}) , & t\to -\infty
\\
2  + O(e^{-2\kappa t}) , & t\to \infty.
\end{cases}
\end{equation}

Since $g_0$ is a linear non-homogeneous function here, 
the speed-amplitude relation now becomes
\begin{equation}
\dot X = \lambda (3-2A)
\end{equation}
which has the feature that $\dot X >0$ when $0<A<\tfrac{3}{2}$
whereas $\dot X <0$ when $A>\tfrac{3}{2}$. 
Thus, 
the speed will change sign when the amplitude passes through the value $\tfrac{3}{2}$. 
From the amplitude expression \eqref{ex3-A},
we find that $A=\tfrac{3}{2}$ occurs at the time $t=-\kappa^{-1}\ln(\sqrt{3})=t^*$.

By applying Proposition~\ref{prop:soln}, 
we obtain the position expression 
\begin{equation}\label{ex3-X}
X= \lambda (t_0-t) - 2(\lambda/\kappa)\ln\big(1+\sqrt{1+e^{2\kappa (t_0 -t)}}\big) +X_0, 
\quad
X_0=\const.
\end{equation}
At the turn-around time $t=t^*$, 
the position is $X(t^*) = -3\lambda\kappa^{-1}\ln(\sqrt{3})$. 
Asymptotically, the position has the behaviour
\begin{equation}\label{ex3-asympt-X}
\begin{aligned}
X -\lambda t & \sim X_0 -\lambda t_0 + O(e^{-\kappa |t|}) , 
\quad
t\to -\infty
\\
X +\lambda t & \sim X_0 + \lambda t_0 -(\lambda/\kappa)\ln(4)+ O(e^{-2\kappa t}), 
\quad
t\to \infty.
\end{aligned}
\end{equation}

The asymptotic behaviour \eqref{ex3-asympt-A} and \eqref{ex3-asympt-X}
describes a travelling-wave peakon \eqref{asympt-peakon} 
with the amplitude $a_{-}=1$, speed $c_{-}=\lambda$, 
and position shift $x_{-}=X_0-c_{-}t_0$
as $t\to -\infty$, 
and with a different amplitude $a_{+}=2$, opposite speed $c_{+}=-\lambda$, 
and position shift $x_{+}=X_0-c_{+}t_0 -\frac{\lambda}{\kappa}\ln(4)$
as $t\to \infty$. 

Therefore, this solution \eqref{ex3-A}, \eqref{ex3-X} is a dynamical peakon 
that evolves from a travelling-wave peakon in the asymptotic past, 
decelerates and reverses direction at the finite time $t^*=-\kappa^{-1}\ln(\sqrt{3})$,
accelerates and then evolves to a different travelling-wave peakon in the asymptotic future. 
In particular, 
the asymptotic amplitude and speed of the peakon change by 
$\Delta a = a_{+}-a_{-} = 1$ and $\Delta c = c_{+}-c_{-} = -2\lambda$. 
See Figure~\ref{fig:direction-reversing-peakon}. 

\begin{figure}[h]
\includegraphics[width=0.6\textwidth]{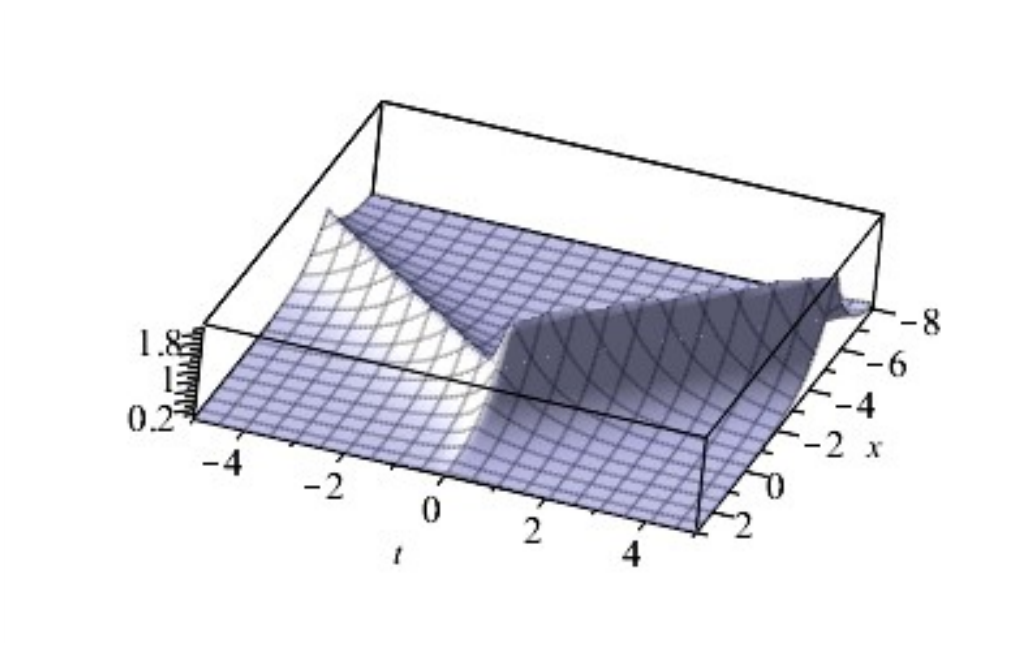}
\caption{Direction-reversing peakon}
\label{fig:direction-reversing-peakon}
\end{figure}

The same kind of asymptotic behaviour can be found in more general wave equations:
\begin{equation}\label{ex3a}
m_t + \kappa(u-a_1)(u-a_2)u^p m +\lambda ((b-u)u^q m)_x=0,
\quad
a_2>b>a_1,
\quad
p,q>0
\end{equation}
and
\begin{equation}\label{ex3b}
m_t + \kappa(u-a_1)(u-a_2)u^p m +\lambda ((u-b_1)(u-b_2))u^q m)_x=0,
\quad
b_2>a_2>b_1>a_1,
\quad
p,q>0
\end{equation}
where $a_1a_2>0$ in both families.

\subsection{Dissipating peakons and blowing up peakons}
\label{dissipative-blowup-peakons}

Dynamical peaked waves that evolve from asymptotic travelling-wave peakons in the asymptotic past
can also exhibit dissipating behaviour in the asymptotic future or blowing-up behaviour in a finite time. 

We consider a family of quadratic nonlinear wave equations
\begin{equation}\label{ex5}
m_t + \kappa (a-u) m +\lambda (u m)_x=0
\end{equation}
where $\kappa,\lambda,a$ are non-zero constants. 
For this family, we have $f=\kappa (a-u)=f^+$ and $g=\lambda u=g^+$,
so thus $F= \kappa (a-u) u_x=F^-$ and $G=\lambda u u_x=G^-$. 
This yields $f_0 = \kappa (a-u)$ and $g_0= \lambda u$ from relation \eqref{f0-F-g0-G}. 

To obtain the dynamical peaked wave solutions \eqref{1peakon}, 
we use Proposition~\ref{prop:soln}. 
The general solution has two different branches, 
one which is smooth for all $t$, 
and one which is has a blow-up at a finite time $t=t^*$. 
Their asymptotic behaviour depends on the sign of $\kappa a$. 
We will proceed by taking $\kappa a>0$. 

First, the smooth solutions are given by 
\begin{equation}\label{AX:ex5-smooth}
A=a/(1+e^{\kappa a(t- t_0)}),
\quad
X= (\lambda/\kappa) \ln\big(2/(1+e^{-\kappa a(t-t_0)})\big) +X_0,
\quad
t_0,X_0=\const
\end{equation}
where $A(t_0)=a/2$ and $X(t_0)=X_0$. 
As $t\to -\infty$, the amplitude and position have the asymptotic behaviour
\begin{equation}
A \sim a + O(e^{-\kappa a|t|}), 
\quad
X -a\lambda t \sim X_0 -a\lambda t_0 +(\lambda/\kappa)\ln(2) + O(e^{-\kappa a|t|}) , 
\end{equation}
which describes a travelling-wave peakon \eqref{asympt-peakon},
with amplitude $a$, speed $c=\lambda a$, 
and position shift $x_0=X_0 -a\lambda t_0 +\frac{\lambda}{\kappa}\ln(2)$. 
This speed-amplitude relation actually holds at all times $t$, 
since $\dot X = g_0(A) =\lambda A$,
as shown by the evolution ODEs \eqref{ODEs-fg}. 

More interestingly, 
as $t\to \infty$ in this solution \eqref{AX:ex5-smooth}, 
the amplitude and its time derivative go to zero. 
Consequently, due to the speed-amplitude relation,
the speed and acceleration also go to zero, 
while 
\begin{equation}
X \sim X_0 +(\lambda/\kappa)\ln(2) + O(e^{-\kappa a|t|})
\end{equation}
remains finite. 
Hence, the peakon undergoes asymptotic braking, 
and reaches the position 
$X_\infty = X_0 +\frac{\lambda}{\kappa}\ln(2)$,
while the amplitude goes extinct. 

This smooth solution \eqref{AX:ex5-smooth} 
therefore describes a dynamical decelerating peakon
that evolves from a travelling-wave peakon in the asymptotic past 
to a dissipating peakon that goes extinct at a finite position in the asymptotic future. 
See Figure~\ref{fig:dissipating-peakon}. 

\begin{figure}[h]
\includegraphics[width=0.6\textwidth]{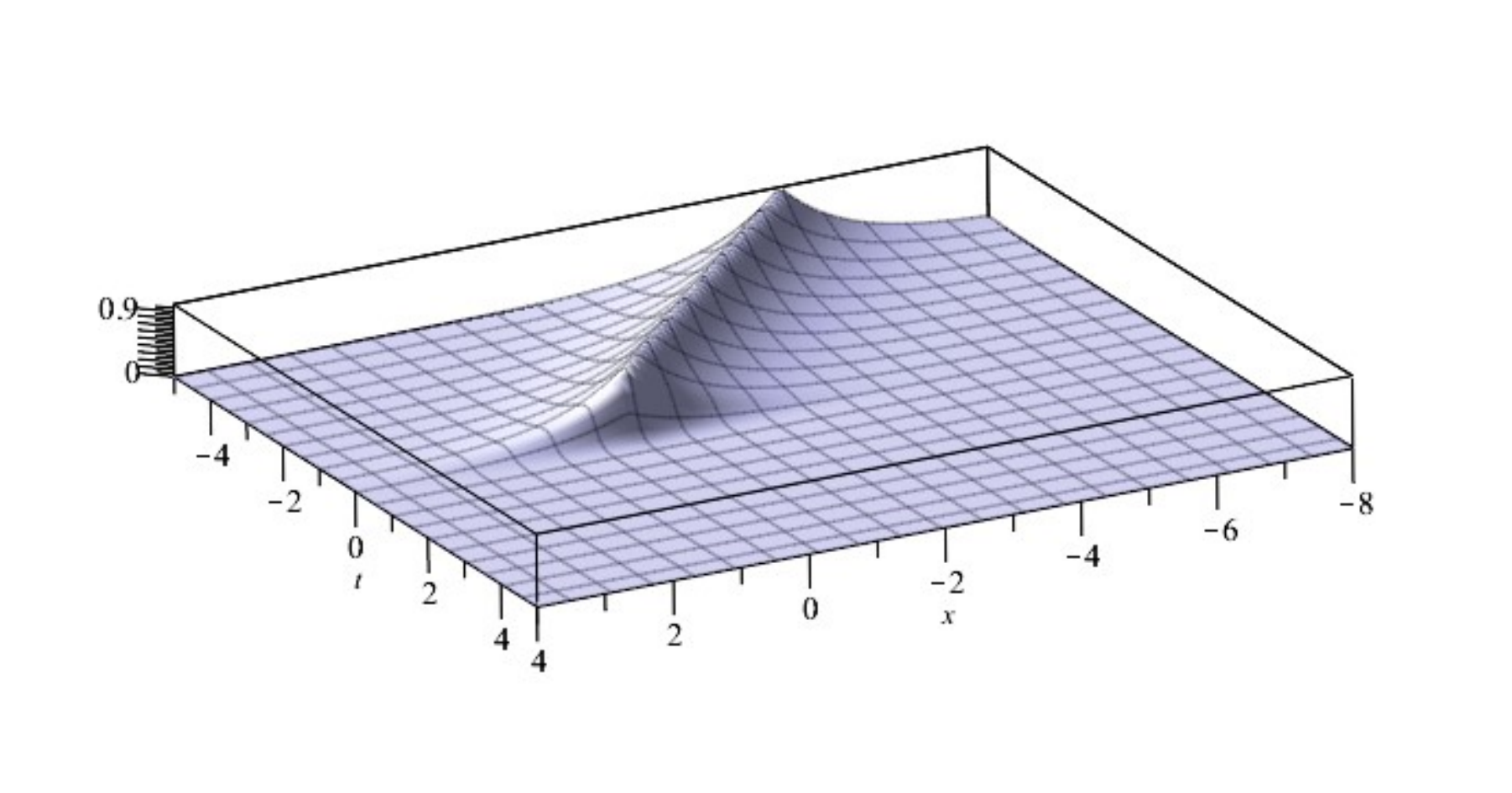}
\caption{Dissipating peakon}
\label{fig:dissipating-peakon}
\end{figure}

We remark that this behaviour can be reversed in time by changing the sign of $\kappa a$ from positive to negative. 
The resulting solutions then describe a dynamical decelerating peakon 
that begins from extinction at a finite position in the asymptotic past, 
and evolves as an anti-dissipating peakon that becomes a travelling-wave peakon in the asymptotic future. 

Next, the blow-up solutions are given by 
\begin{equation}\label{AX:ex5-blowup}
A=a/(1-e^{\kappa a(t- t_0)}),
\quad
X= -(\lambda/\kappa)\ln\big(1-e^{-\kappa a(t-t_0)}\big) +X_0,
\quad
t_0,X_0=\const.
\end{equation}
As $t\to -\infty$, the amplitude and position have the asymptotic behaviour
\begin{equation}
A \sim a + O(e^{-\kappa a|t|}), 
\quad
X -a\lambda t \sim X_0 -a\lambda t_0 + O(e^{-\kappa a|t|}) , 
\end{equation}
which describes a travelling-wave peakon \eqref{asympt-peakon},
with amplitude $a$, speed $c=\lambda a$, and position shift $x_0=X_0 -a\lambda t_0$. 
When $t$ reaches $t^*=t_0<\infty$, 
both the amplitude and the position become unbounded,
$A\to \infty$ and $X\to \infty$ as $t\to t^*$. 
Moreover, the speed and the acceleration also become unbounded, 
$\dot X,\ddot X\to \infty$ as $t\to t^*$,
due to the speed-amplitude relation $\dot X = g_0(A) =\lambda A$. 
Hence, the peakon undergoes a runaway blow-up in a finite time. 

This solution \eqref{AX:ex5-blowup} 
therefore describes a dynamical accelerating peakon
that evolves from a travelling-wave peakon in the asymptotic past 
to a runaway blow-up in a finite time. 
See Figure~\ref{fig:blowing-up-peakon}. 

\begin{figure}[h]
\includegraphics[width=0.5\textwidth]{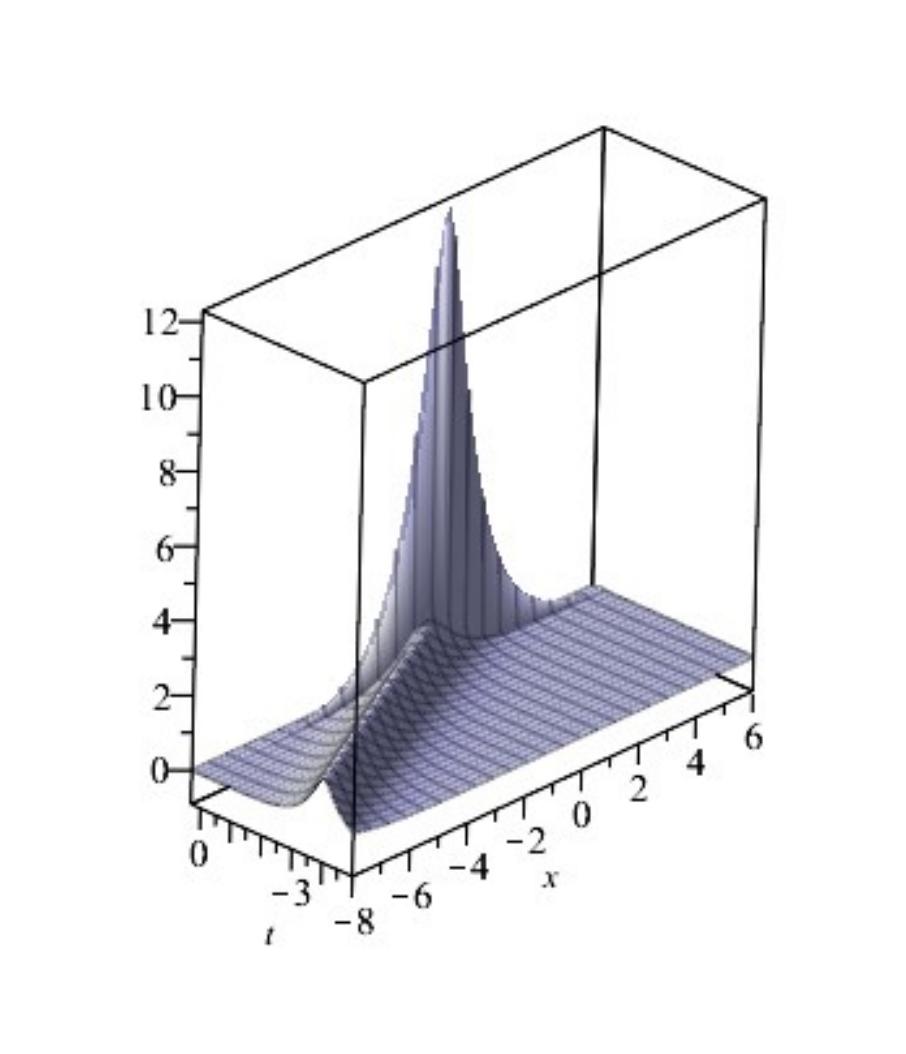}
\caption{Blowing up peakon}
\label{fig:blowing-up-peakon}
\end{figure}

Similar behaviour can be shown to arise
for wave equations with higher-power nonlinearities:
\begin{equation}\label{ex5a}
m_t + \kappa (a-u)u^p m +\lambda (u^{q+1} m)_x=0,
\quad
p,q >0.
\end{equation}

\subsection{Peakon breathers}
\label{peakon-breathers}

A breather is a dynamical wave whose amplitude at each point $x$ is oscillatory in $t$. 
Peakon breathers have been found recently for a new NLS-type peakon equation \cite{AncMob}
$im_t +\tfrac{1}{2}(|u|^2-|u_x|^2)m +i(\Im(\bar u u_x)m)_x=0$
by considering a stationary peakon modified by an oscillatory phase, 
$u(x,t) = a e^{i\omega t} e^{-|x|}$, 
where the frequency is given by $\omega=\tfrac{1}{3}a^3$ in terms of the peakon amplitude $a$. 

Here we will construct a different type of peakon breather 
which is given by a dynamical peaked wave solution 
to certain nonlinear dispersive wave equations contained in the general family \eqref{fg-fam}. 

We start from the expression for a travelling-wave peakon 
$a e^{-|x-ct-x_0|}$ 
and multiply it by the oscillatory factor $\cos(\kappa t)$ to get 
\begin{equation}\label{breather}
u(x,t) = a\cos(\kappa t)  e^{-|x-ct -x_0|}
\end{equation}
where $c$ and $\kappa\neq0$ are constants. 
This yields a peakon breather with an oscillation frequency $\kappa$ 
which is independent of the peak amplitude $a$. 
The speed $c$ is also independent of the peak amplitude $a$ 
and can be positive, negative, or zero. 
See Figure~\ref{fig:peakon-breather}. 

\begin{figure}[h]
\includegraphics[width=0.6\textwidth]{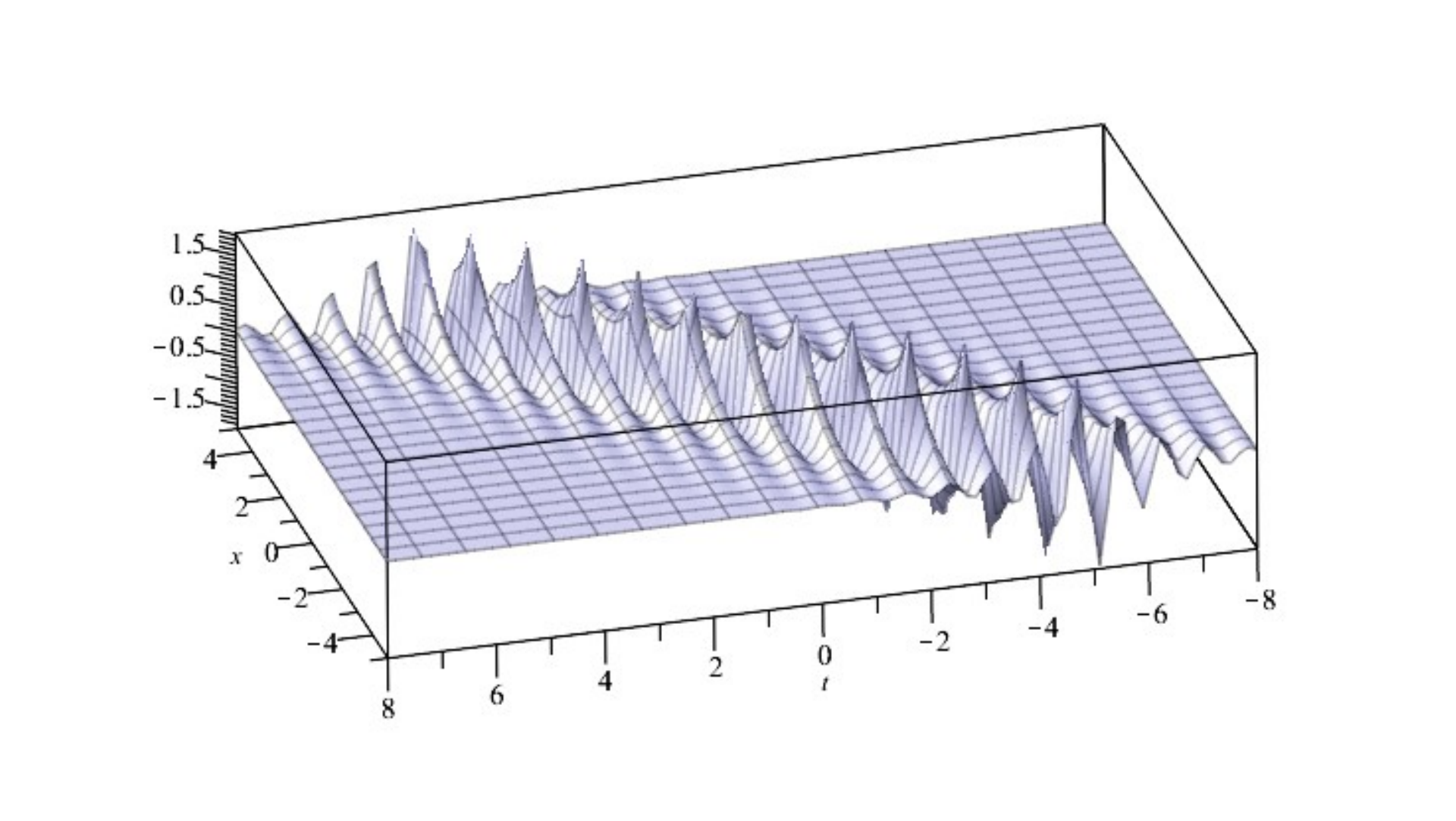}
\caption{Peakon breather}
\label{fig:peakon-breather}
\end{figure}

To proceed, we now determine the nonlinearities $f(u,u_x)$ and $g(u,u_x)$
such that a wave equation \eqref{fg-fam} admits this peakon breather \eqref{breather}
as a dynamical peaked wave solution \eqref{1peakon}. 
Matching the breather expression \eqref{breather} to the general solution expression \eqref{1peakon}, 
we get $A(t) =a\cos(\kappa t)$ and $X(t)=ct+x_0$,
from which we find 
$\dot A = -a\kappa \sin(\kappa t) = -\kappa (a^2-A^2)^{1/2}$ and $\dot X(t) = c$. 
The evolution ODEs \eqref{ODEs-fg} then yield 
$f_0(u) = \kappa (a^2-u^2)^{1/2}/u =\kappa ((a/u)^2-1)^{1/2}$ and $g_0(u)=c$. 
Hence, we can take $f=f_0$ and $g=g_0$. 

This gives the family of wave equations 
\begin{equation}\label{ex-breather}
m_t + \kappa ((a/u)^2-1)^{1/2} m +c m_x=0,
\quad
a\neq 0
\end{equation}
for which the peakon breather is a dynamical peaked wave solution \eqref{1peakon}. 

More general peakon breathers can be obtained in a similar fashion. 
In particular, we can replace $A(t)=a\cos(\kappa t)$ by any periodic function $\phi(t)$,
with period $T>0$. 
The corresponding nonlinearities are given by 
$f_0(u)= -\phi'(\phi^{-1}(u))/u$ and $g_0(u)=c$,
yielding the wave equation $m_t -(\phi'(\phi^{-1}(u))/u)m + c m_x=0$.

\section{Concluding remarks}
\label{remarks}

We have introduced and studied a novel generalization of peakons,
in which the amplitude and the speed are time-dependent dynamical variables \eqref{1peakon}. 
These dynamical peakons arise for a large class of nonlinear dispersive wave equations
which belong to a general class \eqref{fg-fam} whose nonlinearities are given by 
two arbitrary non-singular functions $f(u,u_x)$ and $g(u,u_x)$. 
All equations in this general class possess $N$-peakon weak solutions, 
where the $N=1$ solutions correspond to either travelling-wave peakons or dynamical peakons. 
We have derived explicit conditions on $f(u,u_x)$ and $g(u,u_x)$ that characterize
each of these two kinds of peakons,
and we have also obtained a simple condition for when a dynamical peakon has non-zero acceleration. 

Through examples, 
we have shown that dynamical accelerating peakons can exhibit a wide variety of interesting behaviour. 
In particular, for their amplitude:
\begin{enumerate}[$\bullet$]
\item 
finite-time blow-up or extinction;
\item 
long-time unboundedness, or extinction, or a finite asymptote;
\end{enumerate}
and for their speed and acceleration:
\begin{enumerate}[$\bullet$]
\item
finite-time braking or runaway or a wheelspin-limit;
\item
long-time braking, runaway, or a finite asymptotic limit;
\item 
direction reversal. 
\end{enumerate}
Combinations of these different behaviours produce a plethora of new kinds of peakons,
including 
\begin{enumerate}[$\bullet$]
\item
asymptotic travelling-wave peakons;
\item
asymptotically dissipating/anti-dissipating peakons; 
\item 
blowing-up peakons;
\item 
peakon breathers. 
\end{enumerate}

Such dynamical behaviour has not been seen previously
in any of the nonlinear dispersive wave equations 
studied to date in the literature.

Several open directions of work remain to be done. 
One direction is understanding interactions of these dynamical peakons given by $N$-peakon solutions \eqref{Npeakon}, 
and determining when interactions produce blow-up and runaway behaviour,
as well as the general asymptotic behaviour when the interactions are non-singular. 
Another direction is studying the Cauchy problem for the underlying nonlinear dispersive wave equations to establish local well-posedness, and finding results on global behaviour,
such as whether wave breaking that occurs for nonlinear dispersive wave equations 
with travelling-wave peakons also occurs for the more general class of nonlinear dispersive wave equations with dynamical accelerating peakons. 
In particular, is there a counterpart of the theorems for soliton equations where 
general initial data evolves asymptotically into a train of solitons?

\section*{Acknowledgements}
S.C.A. is supported by an NSERC grant and thanks the Universidad de C\'adiz for support during a visit in which this work was initiated.

The reviewers are thanked for providing helpful comments and references.

\end{document}